\documentclass[conference]{IEEEtran}
\IEEEoverridecommandlockouts
\usepackage{amsmath,amssymb,amsfonts}
\usepackage{array}
\usepackage[caption=false,font=normalsize,labelfont=sf,textfont=sf]{subfig}
\usepackage{textcomp}
\usepackage{stfloats}
\usepackage{verbatim}
\usepackage{graphicx}
\usepackage{cite}
\usepackage{tabularx}
\usepackage{booktabs}

\def\BibTeX{{\rm B\kern-.05em{\sc i\kern-.025em b}\kern-.08em
    T\kern-.1667em\lower.7ex\hbox{E}\kern-.125emX}}

\usepackage{acronym}
\acrodef{5G}{the fifth generation}
\acrodef{MIMO}{multiple-input multiple-output}
\acrodef{MISO}{multiple-input single-output}
\acrodef{RF}{radio frequency}
\acrodef{BS}{base station}
\acrodef{UE}{user equipment}
\acrodef{LoS}{line-of-sight}
\acrodef{NLoS}{non-line-of-sight}
\acrodef{AoA}{angle-of-arrival}
\acrodef{AoD}{angle-of-departure}
\acrodef{UPA}{uniform planar array}
\acrodef{ARV}{array response vector}
\acrodef{CGV}{channel gain vector}
\acrodef{AGV}{antenna gain vector}
\acrodef{EM}{electromagnetic}
\acrodef{MA}{movable antenna}
\acrodef{3D}{three-dimensional}
\acrodef{AWGN}{additive white Gaussian noise}
\acrodef{ERA}{electromagnetically reconfigurable antenna}
\newtheorem{theorem}{\textbf{Theorem}}

\newtheorem{remark}{\textbf{Remark}}

\newenvironment{proof}{\textit{\textbf{Proof:}}}{\hfill$\square$}
\usepackage{color}

\usepackage{tikz} 
\usepackage[utf8]{inputenc}
\usepackage{pgfplots}
\usepackage{tabu,longtable}
\usepackage{diagbox}
\usepackage{threeparttable}
\usepackage{makecell}
\usepackage{multirow} % enable table multirow
\usepackage[bottom=1.04in,top=0.75in,left=0.625in,right=0.625in]{geometry}
\usepackage{units} 		% use \unit command
\usepackage{bm} 		% for bold symbol
\usepackage{amssymb} 	% for AMS symbols
\newcommand{\TT}{\mathsf{T}}
\newcommand{\HH}{\mathsf{H}}

\newcommand{\av}{{\bf a}}
\newcommand{\bv}{{\bf b}}
\newcommand{\cv}{{\bf c}}
\newcommand{\dv}{{\bf d}}

\newcommand{\fv}{{\bf f}}
\newcommand{\gv}{{\bf g}}
\newcommand{\hv}{{\bf h}}

\newcommand{\kv}{{\bf k}}

\newcommand{\sv}{{\bf s}}

\newcommand{\wv}{{\bf w}}
\newcommand{\vv}{{\bf v}}
\newcommand{\xv}{{\bf x}}

\newcommand{\Am}{{\bf A}}

\newcommand{\Fm}{{\bf F}}

\newcommand{\Vm}{{\bf V}}

\newcommand{\Dt}{{\mathsf D}}

\newcommand{\Tt}{{\mathsf T}}

\newcommand{\alphav}{\hbox{\boldmath$\alpha$}}

\usepackage{amsmath}

\usepackage{algorithm}
\usepackage{algpseudocode}% http://ctan.org/pkg/algorithmicx
\algtext*{EndWhile}% Remove "end while" text
\algtext*{EndIf}% Remove "end if" text
\algtext*{EndFor}

\makeatletter
\algnewcommand{\LineComment}[1]{\Statex \hskip\ALG@thistlm \(\triangleright\) #1}
\makeatother
\usepackage{xcolor}
\usepackage{tcolorbox}
\newcommand{\colorcircled}[2]{%
    \raisebox{-0.5ex}{\tcbox[
        colback=#1, % 底色
        colframe=white, % 边框色
        colupper=white, % 文字颜色
        arc=2.5pt, % 圆角半径
        boxrule=1pt, % 边框宽度
        left=2pt, right=2pt, top=2pt, bottom=2pt, % 控制内边距
        boxsep=0.1pt % 盒子内部间距
    ]{#2}}%
}
\definecolor{RoyalBlue}{RGB}{65, 105, 225}  % 皇家蓝

\begin{document}
\bstctlcite{IEEEexample:BSTcontrol} % Control reference style

\title{Tri-Hybrid Multi-User Precoding Based on Electromagnetically Reconfigurable Antennas\\ \vspace{-0.2em}
\author{
Pinjun Zheng\IEEEauthorrefmark{1}, Yuchen Zhang\IEEEauthorrefmark{2}, Tareq Y. Al-Naffouri\IEEEauthorrefmark{2}, Md. Jahangir Hossain\IEEEauthorrefmark{1}, Anas Chaaban\IEEEauthorrefmark{1}\\
\IEEEauthorrefmark{1}\textit{School of Engineering, University of British Columbia, Kelowna, Canada}\\
\IEEEauthorrefmark{2}\textit{CEMSE, King Abdullah University of Science and Technology, Thuwal, Saudi Arabia}\\
(Email: pinjun.zheng@ubc.ca) \vspace{-0.7em}
} 
%\thanks{Identify applicable funding agency here. If none, delete this.}
}

\maketitle

\begin{abstract}
The tri-hybrid precoding architecture based on electromagnetically reconfigurable antennas (ERAs) is a promising solution for overcoming key limitations in multiple-input multiple-output communication systems. Aiming to further understand its potential, this paper investigates the tri-hybrid multi-user precoding problem using pattern reconfigurable ERAs. To reduce model complexity and improve practicality, we characterize each antenna’s radiation pattern using a spherical harmonics decomposition. While mathematically tractable, this approach may lead to over-optimized patterns that are physically unrealizable. To address this, we introduce a projection step that maps the optimized patterns onto a realizable set. Simulation results demonstrate that spherical harmonics-based radiation pattern optimization significantly enhances sum rate performance. However, after projection onto a realizable set obtained from real ERA hardware, the performance gain is notably reduced or even negligible, underscoring the need for more effective projection techniques and improved reconfigurable antenna hardware.
\end{abstract}
\vspace{-0.7em}
\section{Introduction} 

The rapid evolution of wireless communication systems toward \ac{5G} and beyond necessitates significantly higher data rates, enhanced reliability, and reduced latency~\cite{Dogra2021Survey}. A common approach to achieving these goals is to leverage higher frequency bands, which offer larger available bandwidths. However, these approaches face challenges such as severe path loss and hardware constraints. Beyond purely frequency-based solutions, hardware advancements offer promising complements. \Ac{MIMO} technologies, including massive MIMO and hybrid precoding architectures, have been extensively explored to enhance spectral and energy efficiency~\cite{Larsson2014Massive}. More recently, \ac{ERA} has emerged as a transformative technology, capable of dynamically altering \ac{EM} radiation properties (e.g., radiation pattern, polarization, and frequency) of antennas to optimize signal transmission and reception, ultimately enhancing data rates without requiring additional spectrum resources~\cite{Wang2025Electromagnetically}. 

Despite the potential of \acp{ERA}, the tri-hybrid precoding problem, where analog, digital, and \ac{EM}-domain reconfigurable elements are jointly optimized, remains a key challenge in fully unlocking their capabilities and has been relatively less explored~\cite{Castellanos2025Embracing}. A polarization-reconfigurable antenna array-based communication system is studied in~\cite{Castellanos2024Linear}, while a tri-hybrid precoding solution based on dynamic metasurface antennas is reported in~\cite{Castellanos2023Energy}. However, tri-hybrid precoding solutions leveraging pattern reconfigurable antennas have not been sufficiently investigated yet. By dynamically directing energy toward intended users and mitigating interference, this type of \ac{ERA} enables more efficient power utilization compared to fixed radiation pattern antennas~\cite{Zheng2025Enhanced}. Recent advancements in pattern reconfigurable antenna prototyping have already begun to demonstrate these advantages~\cite{Wang2025Electromagnetically}. A recent theoretical study on pattern reconfigurable antenna-based tri-hybrid beamforming was reported in~\cite{Liu2025Tri}. While the results are promising, the approach used in~\cite{Liu2025Tri} relies on discretizing the radiation pattern in the angular domain, which can lead to an extremely large \ac{EM}-domain channel matrix, especially in the \ac{3D} real space, leading to significant computational challenges. Moreover, the study simulates a half-wavelength spaced antenna array using elements whose physical size exceeds half a wavelength, which may result in overly optimistic performance estimates. To tackle these issues, this paper revisits the problem and makes the following key contributions:
\begin{itemize}
    \item \textbf{Radiation Pattern Modeling:} We model the radiation pattern using spherical harmonics decomposition~\cite{Ying2024Reconfigurable}. By selecting an appropriate truncation length, any \ac{3D} radiation pattern can be efficiently captured using a limited number of harmonic coefficients (e.g., 25 in this paper).
    \item \textbf{Optimization Framework:} We address the formulated tri-hybrid optimization problem using an alternating optimization approach, where each subproblem is solved in closed form, thereby obtaining the optimized antenna radiation patterns, analog precoders, and digital recorders.
    \item \textbf{Realizable Projection:} Finally, we project the optimized radiation patterns onto a set of realizable candidates based on practical array-configured antennas with reconfigurable radiation patterns in~\cite{Wang2025Electromagnetically}, ensuring feasibility in real-world implementations.
\end{itemize}

\section{System and Signal Model}

\begin{figure}[t]
  \centering
  \includegraphics[width=\linewidth]{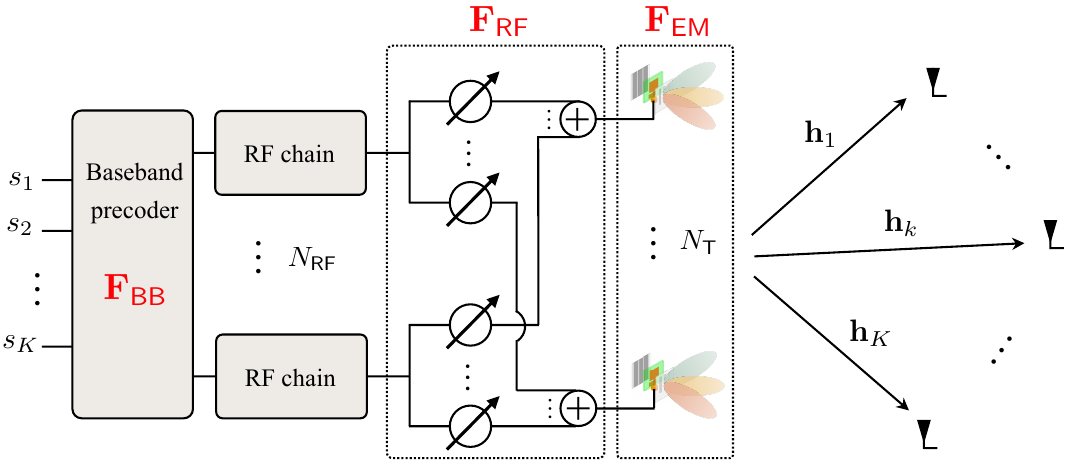}
  \vspace{-2em}
  \caption{
  Illustration of the tri-hybrid precoding architecture for a multi-user communication system utilizing \ac{ERA}.}
  \label{fig_trihybrid}
\end{figure} 

Consider a narrowband downlink multi-user \ac{MISO} communication system where a \ac{BS} equipped with $N_\mathsf{T}$ antennas (each connected to a phase shifter) and $N_\mathsf{RF}$ transmit \ac{RF} chains serves $K$ single-antenna \acp{UE} simultaneously. The \ac{BS} communicates with each \ac{UE} through a single data stream, resulting in a total of $K$ data streams at the \ac{BS}. It is assumed $K\leq N_\mathsf{RF}\leq N_\mathsf{T}$. As shown in Fig.~\ref{fig_trihybrid}, we consider a tri-hybrid preconding architecture.

In this tri-hybrid precoding structure, the \ac{BS} first applies a baseband precoder, $\mathbf{F}_\mathsf{BB}\in\mathbb{C}^{N_\mathsf{RF}\times K}$, to the $K$ data streams digitally (i.e., \emph{digital precoding}). These digitally processed signals are then up-converted to the carrier frequency through $N_\mathsf{RF}$ \ac{RF} chains. Next, a \ac{RF} precoder, $\mathbf{F}_\mathsf{RF}\in\mathbb{C}^{N_\mathsf{T}\times N_\mathsf{RF}}$, is applied using analog phase shifters (i.e., \emph{analog precoding}). Subsequently, these analog-precoded signals are fed into the $N_\mathsf{T}$ antennas, where they are converted into \ac{EM} waves and propagated toward the \acp{UE}. In this work, we explore the utilization of the \ac{ERA} at the \ac{BS}, which enables the electronic adjustment of each antenna's radiation pattern. A recently developed hardware prototype of such an \ac{ERA} has been reported in \cite{Wang2025Electromagnetically}. The antennas at \acp{UE} are assumed to be fixed and isotropic.

Let $s_k\in\mathbb{C}$ denote the data symbol that the \ac{BS} transmits to the $k^\text{th}$ \ac{UE}. We concatenate $\sv=[s_1,s_2,\dots,s_K]^\TT$ as the vector of all data symbols and express $\Fm_\mathsf{BB}=[\fv_{\mathsf{BB},1},\fv_{\mathsf{BB},2},\dots,\fv_{\mathsf{BB},K}]$. Following the described precoding architecture, the signals fed into the $N_\mathsf{T}$ antennas can be written as $\mathbf{x} = \Fm_\mathsf{RF}\Fm_\mathsf{BB}\sv=\sum_{k=1}^K\Fm_\mathsf{RF}\fv_{\mathsf{BB},k}s_k$~\cite{Sohrabi2016Hybrid}.
We assume that $\mathbb{E}\{\sv\sv^\HH\}=\mathbf{I}_{K}$. Let $\hv_k\in\mathbb{C}^{N_\mathsf{T}}$ denote the wireless channel from the \ac{BS} to the $k^\text{th}$ \ac{UE}. The received signal at the $k^\text{th}$ \ac{UE} can be written as
\begin{equation}\label{eq:yk}
    y_k = \hv_k^\TT\xv + z_k = \hv_k^\TT\Fm_\mathsf{RF}\fv_{\mathsf{BB},k}s_k + \hv_k^\TT\sum_{i\neq k}\Fm_\mathsf{RF}\fv_{\mathsf{BB},i}s_i + z_k,
\end{equation}
where $z_k\sim \mathcal{CN}(0,\sigma_k^2)$ is the noise at the $k^\text{th}$ \ac{UE}.

In Section~\ref{sec:CM}, we will show that channel $\hv_k$ can be expressed as $\hv_k=\Fm_\mathsf{EM}^\TT\hv_k^\mathsf{EM}$, $\forall k$, where $\hv_k^{\mathsf{EM}}$ is an effective \ac{EM}-domain channel and $\Fm_\mathsf{EM}\!=\!\mathrm{blkdiag}\big\{\cv^{(1)},\cv^{(2)},\!\dots\!,\cv^{(N_\mathsf{T})}\big\}$ is the \ac{EM} precoder. The vector $\cv^{(n)}$ characterizes the $n^\text{th}$ antenna's radiation pattern.  Assuming perfect knowledge of $\{\hv_k^\mathsf{EM}\}_{k=1}^K$, we can formulate the tri-hybrid multi-user precoding problem that optimizes the total spectral efficiency (sum-rate) as
\begin{subequations}\label{eq:SRmax}
\begin{align}
    \max_{\Fm_\mathsf{EM},\Fm_\mathsf{RF},\Fm_\mathsf{BB}}&\ \sum_{k=1}^K\beta_k\log_2\!\bigg(1+\frac{|(\hv_k^{\mathsf{EM}})^\TT\Fm_\mathsf{EM}\Fm_\mathsf{RF}\fv_{\mathsf{BB},k}|^2}{\sigma^2\!+\!\displaystyle\sum_{i\neq k}|(\hv_k^{\mathsf{EM}})^\TT\Fm_\mathsf{EM}\Fm_\mathsf{RF}\fv_{\mathsf{BB},i}|^2}\bigg)\notag\\
    \mathrm{s.t.\quad }&\quad \mathrm{Tr}(\Fm_\mathsf{RF}\Fm_\mathsf{BB}\Fm_\mathsf{BB}^\HH\Fm_\mathsf{RF}^\HH)\leq P_\mathsf{max},\label{eq:a}\\
    &\quad |\Fm_\mathsf{RF}(i,j)|^2 = {1}/{N_\mathsf{T}},\ \forall i,j,\label{eq:b}\\
    &\quad \Fm_{\mathsf{EM}} = \mathrm{blkdiag}\big\{\cv^{(1)},\cv^{(2)},\dots,\cv^{(N_\mathsf{T})}\big\},\label{eq:c}\\
    &\quad \|\cv^{(n)}\|^2=4\pi,\ n=1,2,\dots,N_\mathsf{T},\label{eq:d}
\end{align}
\end{subequations}
where $\beta_k>0$ is a weight for the $k^\text{th}$ \ac{UE}, $P_\mathsf{max}$ is the maximum transmit power, and $\|\cv^{(n)}\|^2$ calculates the total power of the antenna gain, which is constrained to $4\pi$ based on the definition of antenna gain~\cite[Remark 1]{Zheng2025Enhanced}. Next, we characterize the channel model to incorporate the \ac{EM}-domain reconfigurability demonstrated in~\eqref{eq:SRmax}.

\section{Channel Model}\label{sec:CM}
In general, we can express $\hv_k$ in the frequency domain as
\begin{equation}\label{eq:hk}
    \hv_k = \sqrt{\frac{N_\mathsf{T}}{L_k}} \sum_{\ell=1}^{L_k} \alphav_{k,\ell}\odot\gv_{k,\ell}\odot\av_{k,\ell},
\end{equation}
where $\odot$ denotes Hadamard product, $\alphav_{k,\ell}$ is the complex \ac{CGV}, $\gv_{k,\ell}$ is the \ac{AGV}, and $\av_{k,\ell}$ is the \ac{ARV}, corresponding to the $\ell^\text{th}$ path to the $k^\text{th}$ \ac{UE}. The \ac{CGV} and \ac{AGV} can be expressed as
\begin{align}
    \alphav_{k,\ell} &= \big[\alpha_{k,\ell}^{(1)},\alpha_{k,\ell}^{(2)},\dots,\alpha_{k,\ell}^{(N_\mathsf{T})}\big]^\TT\in\mathbb{C}^{N_\mathsf{T}},\\
    \gv_{k,\ell} &= \big[G_{k,\ell}^{(1)},G_{k,\ell}^{(2)},\dots,G_{k,\ell}^{(N_\mathsf{T})}\big]\in\mathbb{R}_+^{N_\mathsf{T}},\label{eq:AGV}
\end{align}
where $\alpha_{k,\ell}^{(n)}$ denotes the complex channel gain between  the $n^\text{th}$ transmit antenna element and the $k^\text{th}$ \ac{UE} through the $\ell^\text{th}$ path, and $G_{k,\ell}^{(n)}$ denotes the antenna gain of the $n^\text{th}$ transmit antenna element for this path. Note that in the frequency domain, the pulse-shaping filter response can be absorbed into \ac{CGV} and is therefore omitted in this work~\cite{Venugopal2017Venugopal}.

\subsection{Fixed Radiation Pattern Antennas-Based Channel}
We begin our channel model derivation by examining conventional antennas with fixed radiation patterns. In this case, all antenna elements share the same and fixed radiation pattern, which we denote as $G(\theta,\phi)$. Here, the radiation pattern $G$ is a function of \ac{AoD} $(\theta,\phi)$~\cite{Zheng2025Enhanced}, where $\theta$ denotes the inclination angle and $\phi$ denotes the azimuth angle. Since~\eqref{eq:hk} is a general channel expression, the following subsections analyze the far-field and near-field scenarios separately.

\subsubsection{Far-Field}
In far-field scenarios, all antenna elements share the same channel gain and antenna gain, i.e., $\alpha_{k,\ell}^{(n)}=\alpha_{k,\ell}$, $G_{k,\ell}^{(n)}=G(\theta_{k,\ell},\phi_{k,\ell})$, $\forall n$. In addition, assuming an $N_\mathsf{T}^{\mathsf{h}}\times N_\mathsf{T}^{\mathsf{v}}$ \ac{UPA} configuration (thus $N_\mathsf{T}=N_\mathsf{T}^{\mathsf{h}} N_\mathsf{T}^{\mathsf{v}}$) at the \ac{BS}, the \ac{ARV} in the far-field is given by 
\begin{equation}
    \av_{k,\ell} = \frac{1}{\sqrt{N_\mathsf{T}}}e^{-j2\pi\varphi^\mathsf{h}_{k,\ell}\kv(N_\Tt^\mathsf{h})}\!\otimes\! e^{-j2\pi\varphi^\mathsf{v}_{k,\ell}\kv(N_\Tt^\mathsf{v})},
\end{equation}
where~$\kv(N)=[0,1,\dots,N-1]^\TT$, and~$\varphi^\mathsf{h}_{k,\ell}$ and~$\varphi^\mathsf{v}_{k,\ell}$ are the spatial angles corresponding to the horizontal and vertical dimensions, respectively. Assuming the \ac{UPA} is deployed on the $YOZ$-plane of the transmitter's body coordinate system, we obtain $\varphi^\mathsf{h}_{k,\ell}\triangleq{d}\sin(\phi_{k,\ell})\sin(\theta_{k,\ell})/{\lambda}$ and $\varphi^\mathsf{v}_{k,\ell}\triangleq{d}\cos(\theta_{k,\ell})/{\lambda}$, where~$\lambda$ is the wavelength of the operating frequency and $d$ is the inter-element spacing of the transmit antenna array. Therefore, \eqref{eq:hk} is reduced to the commonly used far-field multipath channel model as~\cite{Ayach2014Spatially,Tarboush2021TeraMIMO,Zheng2024Coverage}
\begin{equation}\label{eq:hCV}
    \hv_k = \sqrt{\frac{N_\mathsf{T}}{L_k}}\sum_{\ell=1}^{L_k} \alpha_{k,\ell} G(\theta_{k,\ell},\phi_{k,\ell}) \av(\theta_{k,\ell},\phi_{k,\ell}). 
\end{equation}

\subsubsection{Near-Field}
In near-field scenarios, the channel gains of the propagation through different antenna elements are not identical due to the spherical wavefront~\cite{Yuan2023Spatial}. Thus, $\alpha_{k,\ell}^{(n)}\neq \alpha_{k,\ell}^{(m)}$, $\forall n\neq m$ in general. Moreover, the antenna gains at different elements vary due to the varying \ac{AoD} across the array. Specifically, we can write $G_{k,\ell}^{(n)}=G\big(\theta_{k,\ell}^{(n)},\phi_{k,\ell}^{(n)}\big)$, $\forall n$. Here,~$\big(\theta_{k,\ell}^{(n)},\phi_{k,\ell}^{(n)}\big)$ represents the \ac{AoD} at the $n^\text{th}$ antenna element. In addition, the \ac{ARV} in the near-field is given by~\cite{Lu2024Tutorial}
\begin{equation}
    \av_{k,\ell} = \frac{1}{\sqrt{N_\mathsf{T}}}\Big[e^{-j\frac{2\pi}{\lambda}(r_{k,\ell}-r_{k,\ell}^{(1)})},\dots,e^{-j\frac{2\pi}{\lambda}(r_{k,\ell}-r_{k,\ell}^{(N_\mathsf{T})})} \Big]^\TT,\notag
\end{equation}
where $r_{k,\ell}^{(n)}$ is the propagation distance between the $n^\text{th}$ transmit antenna element and the $k^\text{th}$ \ac{UE} through the $\ell^\text{th}$ path, and $r_{k,\ell}$ is the propagation distance between the $k^\text{th}$ \ac{UE} and a reference point at the \ac{BS}. 

The above analysis demonstrates that~\eqref{eq:hk} is a general channel model applicable to both far-field and near-field scenarios. Therefore, the following derivation will proceed based on~\eqref{eq:hk} without distinguishing between the two. We will incorporate the reconfigurability of each antenna element's radiation pattern into~\eqref{eq:hk} and obtain a concise representation using the \ac{EM} precoder and the \ac{EM}-domain channel.

\subsection{\ac{ERA}-Based Channel}
When considering radiation pattern reconfigurability of \acp{ERA}, \eqref{eq:hk} still holds. However, the radiation pattern of each antenna can be different. Hence, the antenna gain of the $n^\text{th}$ antenna element should be written as
\begin{equation}\label{eq:Gn_kl}
G_{k,\ell}^{(n)}=G^{(n)}\big(\theta_{k,\ell}^{(n)},\phi_{k,\ell}^{(n)}\big).    
\end{equation}
Here, $G^{(n)}$ denote the distinct radiation pattern of the $n^\text{th}$ antenna element, and the \ac{AoD} $\big(\theta_{k,\ell}^{(n)},\phi_{k,\ell}^{(n)}\big)$ can be identical or varying across antenna elements. There are several ways to model element radiation pattern $G^{(n)}$. This paper adopts the spherical harmonics representation, as in~\cite{Ying2024Reconfigurable,Zheng2025Enhanced}, to ensure mathematical traceability.

\subsubsection{Spherical Harmonics Representation}
As shown in~\cite{Zheng2025Enhanced}, any radiation pattern can be decomposed into a superposition of infinitely many spherical harmonics such that
\begin{equation}\label{eq:SH}
    G^{(n)}(\theta,\phi)=\sum_{u=0}^{+\infty}\sum_{q=-u}^{u} c_{u q}^{(n)} Y_u^q(\theta,\phi).
\end{equation}
Here, $c_{u q}^{(n)}$ denotes the harmonic coefficient and $Y_u^q(\theta,\phi)$ is the \emph{real} spherical harmonics defined as
\begin{equation}\label{eq:Ylm}
Y_u^q(\theta,\phi)\!=\! \left\{
\begin{array}{ll}
\!\!\!\sqrt{2}N_u^qP_u^q(\cos{\theta})\cos{(q\phi)}, & q>0,\\
\!\!\!\sqrt{2}N_u^{|q|}P_u^{|q|}(\cos{\theta})\sin{(|q|\phi)}, & q< 0,\\
\!\!\!N_u^0P_u^0(\cos{\theta}), & q=0,
\end{array} \right.
\end{equation}
where $N_u^q=\sqrt{\frac{2u+1}{4\pi}\frac{(u-q)!}{(u+q)!}}$ is a normalization factor, and $P_u^q(\cos{\theta})$ represents the associated Legendre functions of $u^\text{th}$ degree and $q^\text{th}$ order. This set of functions $Y_{u}^q(\theta,\phi)$ constitutes a complete real orthonormal basis on the spherical space. 

We further perform a truncation to~\eqref{eq:SH}. By defining the truncation length $T=U^2+2U+1$, we can approximate the radiation pattern of the $n^\text{th}$ antenna as
\begin{align}\label{eq:gn}
    G^{(n)}(\theta,\phi) \approx \sum_{u=0}^{U}\sum_{q=-u}^u c_{uq}^{(n)} Y_u^q(\theta,\phi) = \sum_{t=1}^T \tilde{c}_t^{(n)} \tilde{Y}_t(\theta,\phi),
\end{align}
where $\tilde{c}_t^{(n)} = c_{uq}^{(n)}$ and $\tilde{Y}_t(\theta,\phi) = Y_u^q(\theta,\phi)$, for $t=u^2+u+q+1$, $u\in[0,U]$, $q\in[-u,u]$. For convenience, we further concatenate $\cv^{(n)} \triangleq[\tilde{c}_1^{(n)},\tilde{c}_2^{(n)},\dots,\tilde{c}_T^{(n)}]^\TT\in\mathbb{R}^T$, 
$\bv(\theta,\phi) \triangleq[\tilde{Y}_1(\theta,\phi),\tilde{Y}_2(\theta,\phi),\dots,\tilde{Y}_T(\theta,\phi)]^\TT\in\mathbb{R}^T$, where $\cv^{(n)}$ is the spherical coefficient vector of the $n^\text{th}$ antenna element, and $\bv(\theta,\phi)$ is the spherical basis vector. Using this truncated version of spherical harmonics, we can rewrite~\eqref{eq:SH} as
\begin{equation}\label{eq:Gn_thetaphi}
    G^{(n)}(\theta,\phi) \approx \bv(\theta,\phi)^\TT\cv^{(n)}.
\end{equation}

\subsubsection{EM-Domain Channel Modeling}
According to~\eqref{eq:Gn_kl} and~\eqref{eq:Gn_thetaphi}, we can express entries in the \ac{AGV}~\eqref{eq:AGV} as $G_{k,\ell}^{(n)} = \bv\big(\theta_{k,\ell}^{(n)},\phi_{k,\ell}^{(n)}\big)^\TT\cv^{(n)}$. By defining $\Fm_{\mathsf{EM}}\in\mathbb{R}^{TN_\mathsf{T}\times N_\mathsf{T}}$ and $\dv_{k,\ell}\in\mathbb{R}^{TN_\mathsf{T}}$ as
\begin{align}
    \Fm_{\mathsf{EM}} &\triangleq \mathrm{blkdiag}\big\{\cv^{(1)},\cv^{(2)},\dots,\cv^{(N_\mathsf{T})}\big\},\notag \\
    \dv_{k,\ell} &\triangleq \big[\bv\big(\theta_{k,\ell}^{(1)},\phi_{k,\ell}^{(1)}\big)^\TT\!\!,\bv\big(\theta_{k,\ell}^{(2)},\phi_{k,\ell}^{(2)}\big)^\TT\!\!,\dots,\bv\big(\theta_{k,\ell}^{(N_\mathsf{T})},\phi_{k,\ell}^{(N_\mathsf{T})}\big)^\TT\big]^\TT,\notag
\end{align}
we can express the single path component in~\eqref{eq:hk} as
\begin{align}
    \alphav_{k,\ell}\odot\gv_{k,\ell}\odot\av_{k,\ell} = \Fm_\mathsf{EM}^\TT\underbrace{\Big(\dv_{k,\ell}\odot\big((\alphav_{k,\ell}\odot\av_{k,\ell})\otimes\mathbf{1}_T\big)\Big)}_{\hv_{k,\ell}^\mathsf{EM}\in\mathbb{C}^{TN_\mathsf{T}}}.\notag
\end{align}
Note that $\hv_{k,\ell}^\mathsf{EM}=\dv_{k,\ell}\odot\big((\alphav_{k,\ell}\odot\av_{k,\ell})\otimes\mathbf{1}_T\big)$ is the \ac{EM}-domain channel of the path $(k,\ell)$. Therefore, we can rewrite~\eqref{eq:hk} as
\begin{align}
    \hv_k = \sqrt{\frac{N_\mathsf{T}}{L_k}} \sum_{\ell=1}^{L_k} \Fm_\mathsf{EM}^\TT\hv_{k,\ell}^\mathsf{EM} &= \Fm_\mathsf{EM}^\TT\underbrace{\sqrt{\frac{N_\mathsf{T}}{L_k}} \sum_{\ell=1}^{L_k} \hv_{k,\ell}^\mathsf{EM}}_{\hv_k^\mathsf{EM}\in\mathbb{C}^{TN_\mathsf{T}}}= \Fm_\mathsf{EM}^\TT\hv_k^\mathsf{EM},\notag
\end{align}
where $\hv_k^\mathsf{EM}=\sqrt{\frac{N_\mathsf{T}}{L_k}} \sum_{\ell=1}^{L_k} \hv_{k,\ell}^\mathsf{EM}$ is defined as the \emph{\ac{EM}-domain channel} and $\Fm_\mathsf{EM}$ is defined as the \emph{\ac{EM} precoder}. Substituting $\hv_k=\Fm_\mathsf{EM}^\TT\hv_k^\mathsf{EM}$ into~\eqref{eq:yk} yields
\begin{equation}\label{eq:box}
\boxed{
\begin{aligned}
   y_k = &\big(\hv_k^{\mathsf{EM}}\big)^\TT\Fm_\mathsf{EM}\Fm_\mathsf{RF}\fv_{\mathsf{BB},k}s_k \\
   &\qquad\qquad\quad + \big(\hv_k^{\mathsf{EM}}\big)^\TT\sum_{i\neq k}\Fm_\mathsf{EM}\Fm_\mathsf{RF}\fv_{\mathsf{BB},i}s_i + z_k.
\end{aligned}
}
\end{equation}
This formula justifies the optimization problem described in~\eqref{eq:SRmax}.

\begin{remark}
    This tri-hybrid formulation indicates that, in addition to digital and analog precoding, one can further reconfigure the radiation pattern of each antenna element by adjusting their spherical harmonics coefficients $\{\cv^{(n)}\}_{n=1}^{N_\mathsf{T}}$, i.e., the EM precoder $\Fm_\mathsf{EM}$. While this approach offers preferred traceability compared to the model in~\cite{Wang2025Electromagnetically} that selects a single radiation pattern for each antenna element from a set of candidates, this model may suffer from the over-relaxation issue. Optimization over $\Fm_\mathsf{EM}$ (i.e., $\{\cv^{(n)}\}_{n=1}^{N_\mathsf{T}}$) may result in highly optimized yet physically unrealizable radiation patterns. To address this problem, a projection step can be incorporated after the optimization process, projecting the optimized arbitrary radiation patterns to the feasible radiation pattern candidate set.  
\end{remark}

\section{A WMMSE-Based Precoding Solution}
According to the well-known weighted minimum mean square error (WMMSE) solution to the weighted sum-rate maximization problem~\cite{Shi2011Iteratively}, the maximization problem~\eqref{eq:SRmax} is equivalent to the following minimization problem: 
\begin{subequations}\label{eq:WMMSE}
    \begin{align}
\min_{\wv,\vv,\Fm_\mathsf{EM},\Fm_\mathsf{RF},\Fm_\mathsf{BB}}&\quad \sum_{k=1}^K\beta_k(w_ke_k-\ln w_k)\\
        \mathrm{s.t.\quad\quad }&\quad \eqref{eq:a},\eqref{eq:b},\eqref{eq:c},\eqref{eq:d}.
    \end{align}
\end{subequations}
Here, $\wv=[w_1,w_2,\dots,w_K]^\TT\in\mathbb{R}_+^K$ represents a set of auxiliary variables, and $\vv=[v_1,v_2,\dots,v_K]^\TT\in\mathbb{C}^K$ is a collection of fictitious linear combiners as in~\cite{Shi2011Iteratively} such that
\begin{multline}
     \hspace{-1em} e_k \!=\! \big(1-v_k(\hv_k^\mathsf{EM})^\TT\Fm_\mathsf{EM}\Fm_\mathsf{RF}\fv_{\mathsf{BB},k}\!\big)\!\big(1-v_k(\hv_k^\mathsf{EM})^\TT\Fm_\mathsf{EM}\Fm_\mathsf{RF}\fv_{\mathsf{BB},k}\!\big)^*\!\!\\
     +|v_k|^2(\hv_k^\mathsf{EM})^\TT\!\Big(\!\sum_{i\neq k}\Fm_\mathsf{EM}\Fm_\mathsf{RF}\fv_{\mathsf{BB},i}\fv_{\mathsf{BB},i}^\HH\Fm_\mathsf{RF}^\HH\Fm_\mathsf{EM}^\HH\!\Big)(\hv_k^\mathsf{EM})^*
     + \sigma_k^2|v_k|^2\!.\notag
\end{multline}
The advantage of~\eqref{eq:WMMSE} is that the optimization of each variable in $\{\wv,\vv,\Fm_\mathsf{EM},\Fm_\mathsf{RF},\Fm_\mathsf{BB}\}$ while fixing others is convex, whereas~\eqref{eq:SRmax} lacks this property.

To solve~\eqref{eq:WMMSE}, we employ the block coordinate descent method, similar to~\cite{Shi2011Iteratively}, to optimize the unknown variables alternately. During these alternating updates, we focus on optimizing the combined effect of the analog and digital precorders instead of optimizing them individually to reduce algorithm complexity. Specifically, we introduce a fictitious fully digital precoder, defined as $\Fm_\Dt\triangleq\Fm_\mathsf{RF}\Fm_\mathsf{BB}$. Then, we optimize $\{\wv,\vv,\Fm_\mathsf{EM},\Fm_\mathsf{D}\}$ alternately. Afterward, we decompose the optimized $\Fm_\Dt$ to obtain the optimal $\Fm_\mathsf{RF}$ and $\Fm_\mathsf{BB}$.
A summary of this optimization solution is provided in Algorithm~\ref{algo:1}. In the following subsections, we provide a detailed solution for each of the five steps highlighted in Algorithm~\ref{algo:1}.

 \begin{algorithm}[t]
 \caption{An Alternating Solver for~\eqref{eq:WMMSE}}
 \label{algo:1}
 \begin{algorithmic}[1]
 \State \textbf{Input:}  $\{{\hv}_k^\mathsf{EM}\}_{k=1}^K$.\qquad \textbf{Output:}  $\Fm_\mathsf{EM}^{\mathsf{opt}},\Fm_\mathsf{RF}^{\mathsf{opt}},\Fm_\mathsf{BB}^{\mathsf{opt}}$.
\State Set iteration index $i=0$, initialize $\wv^{[0]}$,$\vv^{[0]}$,$\Fm_\mathsf{EM}^{[0]}$,$\Fm_\mathsf{RF}^{[0]}$,$\Fm_\mathsf{BB}^{[0]}$, and compute $\Fm_\mathsf{D}^{[0]}=\Fm_\mathsf{RF}^{[0]}\Fm_\mathsf{BB}^{[0]}$.
\Statex{\hspace{-1.2em}\textbf{Repeat:}} 

\Statex{\textbf{Step} \colorcircled{RoyalBlue}{\footnotesize{1}}:} Update $\vv^{[i+1]}$ given $\{\wv^{[i]},\Fm_\mathsf{EM}^{[i]},\Fm_\mathsf{D}^{[i]}\}$.
\Statex{\textbf{Step} \colorcircled{RoyalBlue}{\footnotesize{2}}:} Update $\wv^{[i+1]}$ given $\{\vv^{[i+1]},\Fm_\mathsf{EM}^{[i]},\Fm_\mathsf{D}^{[i]}\}$.
\Statex{\textbf{Step} \colorcircled{RoyalBlue}{\footnotesize{3}}:} Update $\Fm_\mathsf{D}^{[i+1]}$ given $\{\wv^{[i+1]},\vv^{[i+1]},\Fm_\mathsf{EM}^{[i]}\}$.
\Statex{\textbf{Step} \colorcircled{RoyalBlue}{\footnotesize{4}}:} Update $\Fm_\mathsf{EM}^{[i+1]}$ given $\{\wv^{[i+1]},\vv^{[i+1]},\Fm_\mathsf{D}^{[i+1]}\}$.
\State $i = i + 1$.

\Statex{\hspace{-1.2em}\textbf{Until}} convergence.

\Statex{\hspace{-1.2em}\textbf{Step} \colorcircled{RoyalBlue}{\footnotesize{5}}:} Decompose $\{\Fm_\mathsf{RF}^{[i+1]},\Fm_\mathsf{BB}^{[i+1]}\}=\arg\displaystyle\min_{\Fm_\mathsf{RF},\Fm_\mathsf{BB}}\ \|\Fm_\mathsf{D}^{[i+1]}-\Fm_\mathsf{RF}\Fm_\mathsf{BB}\|_\mathsf{F}^2$, s.t.~\eqref{eq:a}, \eqref{eq:b}.

\Statex{\hspace{-1.2em}\textbf{Return}} $\Fm_\mathsf{EM}^{[i+1]},\Fm_\mathsf{RF}^{[i+1]},\Fm_\mathsf{BB}^{[i+1]}$.

\end{algorithmic} 
\end{algorithm}

\subsection{Solutions to Step \colorcircled{RoyalBlue}{\textnormal{\footnotesize{1}}}, Step \colorcircled{RoyalBlue}{\textnormal{\footnotesize{2}}}, and Step \colorcircled{RoyalBlue}{\textnormal{\footnotesize{3}}} }
Following the principle in~\cite{Shi2011Iteratively}, one can find closed-form solutions for Step \colorcircled{RoyalBlue}{\textnormal{\footnotesize{1}}}, \colorcircled{RoyalBlue}{\textnormal{\footnotesize{2}}}, and \colorcircled{RoyalBlue}{\textnormal{\footnotesize{3}}} in Algorithm~\ref{algo:1} as:
\begin{align}
    v_k^\mathsf{opt} &= \frac{\fv_{\mathsf{D},k}^\HH\Fm_\mathsf{EM}^\HH(\hv_k^{\mathsf{EM}})^*}{(\hv_k^{\mathsf{EM}})^\TT\Big(\displaystyle\sum_{i=1}^K\Fm_\mathsf{EM}\fv_{\Dt,i}\fv_{\Dt,i}^\HH\Fm_\mathsf{EM}^\HH\Big)(\hv_k^{\mathsf{EM}})^* + \sigma_k^2},\\
    w_k^\mathsf{opt} &= e_k^{-1} = \big(1-v_k^\mathsf{opt}(\hv_k^{\mathsf{EM}})^\TT\Fm_\mathsf{EM}\fv_{\mathsf{D},k}\big)^{-1},\\
    \fv_{\Dt,k}^\mathsf{opt} &= \Big(\sum_{j=1}^K\beta_j w_j^{\mathsf{opt}}\big|v_j^{\mathsf{opt}}\big|^2 \Fm_\mathsf{EM}^\HH(\hv_j^{\mathsf{EM}})^*(\hv_j^{\mathsf{EM}})^\TT\Fm_\mathsf{EM} + \lambda_k\mathbf{I}\Big)^{-1} \notag\\
    &\qquad\qquad\qquad\qquad\quad \times \beta_k w_k^\mathsf{opt} (v_k^\mathsf{opt})^* \Fm_\mathsf{EM}^\HH (\hv_k^{\mathsf{EM}})^*,
\end{align}
where $[\fv_{\mathsf{D},1},\fv_{\mathsf{D},2},\dots,\fv_{\mathsf{D},K}]=\Fm_\Dt$, $\lambda_k\geq 0 $ is a Lagrange multiplier attached to the total transmit power constraint~\eqref{eq:a}. The value of $\lambda_k$ should be chosen to satisfy the complementarity slackness condition of the inequality constraint~\eqref{eq:a}. It can be determined through a simple bisection search since the transmit power as a function of $\lambda_k$ here is monotonic.

\subsection{Solution to Step \colorcircled{RoyalBlue}{\textnormal{\footnotesize{4}}} }
Recalling~\eqref{eq:c}, the \ac{EM} precoder $\Fm_\mathsf{EM}$ is a collection of the spherical harmonic coefficient vectors of all antenna radiation patterns, $\{\cv^{(n)}\}_{n=1}^{N_\mathsf{T}}$, where each $\cv^{(n)}$ is subject to an equality constraint~\eqref{eq:d}. To obtain closed-form optimum, we solve Step~\colorcircled{RoyalBlue}{\textnormal{\footnotesize{4}}} by updating each of $\{\cv^{(n)}\}_{n=1}^{N_\mathsf{T}}$ alternately.

\subsubsection{Problem Reformulation}
Before solving this optimization problem, we perform an additional analysis of $\cv^{(n)}$. Recalling~\eqref{eq:Gn_thetaphi}, we note that there is an implicit constraint on $\cv^{(n)}$ that the synthesized radiation pattern needs to be strictly positive, as demonstrated in~\eqref{eq:AGV}. However, neither~\eqref{eq:SRmax} nor~\eqref{eq:WMMSE} includes this constraint, potentially leading to an optimized radiation pattern with non-positive values, which lack physical feasibility. By further examining~\eqref{eq:Gn_thetaphi}, we observe that the first entry in $\bv(\theta,\phi)$ is the $0^\text{th}$-order spherical harmonics $Y_0^0(\theta,\phi)$, which is a sphere with a constant positive value across all $(\theta,\phi)$. Therefore, under the power constraint~\eqref{eq:d}, we can always assign a sufficiently large coefficient to $Y_0^0(\theta,\phi)$ to ensure a strictly positive radiation pattern throughout the optimization process. To do so, we first perform the following partitions:
\begin{equation}\label{eq:partition}
   \hv_k^\mathsf{EM} \!=\! \begin{bmatrix}
        \hv_{k,(1)}^\mathsf{EM}\vspace{-0.2em}\\
        \vdots\\
        \hv_{k,(N_\mathsf{T})}^\mathsf{EM}
    \end{bmatrix},\ 
    \hv_{k,(n)}^\mathsf{EM} \!=\! \begin{bmatrix}
        h_{k,(n)}^\mathsf{DC}\vspace{0.2em}\\
        \hv_{k,(n)}^\mathsf{AC}
    \end{bmatrix},\ 
    \cv^{(n)} \!=\! \begin{bmatrix}
        c_{(n)}^\mathsf{DC}\vspace{0.2em}\\
        \cv_{(n)}^\mathsf{AC}
    \end{bmatrix},
\end{equation}
where $\hv_{k,(n)}^\mathsf{EM}\in\mathbb{C}^{T}$, $\hv_{k,(n)}^\mathsf{AC}\in\mathbb{C}^{T-1}$, and $\cv_{(n)}^\mathsf{AC}\in\mathbb{R}^{T-1}$. Hence, we can fix a (large enough) constant $0^\text{th}$-order coefficient $c_{(n)}^\mathsf{DC}=\eta\in(0,\sqrt{4\pi})$, $\forall n$, and optimize the other coefficients $\cv_{(n)}^\mathsf{AC}$ with the constraint $\|\cv_{(n)}^\mathsf{AC}\|^2=4\pi-\eta^2$.

Therefore, in this step, we solve the following problem:
\begin{subequations}\label{eq:WMMSE2}
    \begin{align}
\min_{\cv_{(n)}^\mathsf{AC}}&\quad \sum_{k=1}^K\beta_k(w_ke_k-\ln w_k)\\
        \mathrm{s.t.}&\quad \|\cv_{(n)}^\mathsf{AC}\|^2=4\pi-\eta^2, \label{eq:WMMSE2b}
    \end{align}
\end{subequations}
given $\wv,\vv,\Fm_\mathsf{D}$, and $\{\cv_{(j)}^\mathsf{AC}\}_{j\neq n}$.
Now, we derive $e_k$ as a function of $\cv_{(n)}^\mathsf{AC}$. Based on~\eqref{eq:partition}, we can express a component in $e_k$ as
\begin{align}
    &(\hv_k^{\mathsf{EM}})^\TT\Fm_\mathsf{EM}\Fm_\mathsf{RF}\fv_{\mathsf{BB},i} = (\hv_k^{\mathsf{EM}})^\TT\Fm_\mathsf{EM}\fv_{\mathsf{D},i}\\
    =& \sum_{n=1}^{N_\mathsf{T}}\Big((\hv_{k,(n)}^\mathsf{EM})^\TT\cv^{(n)}\Big) f_{\Dt,i}^{(n)}\\
    =& \sum_{n=1}^{N_\mathsf{T}}\Big( (\hv_{k,(n)}^\mathsf{AC})^\TT\cv_{(n)}^{\mathsf{AC}} + h_{k,(n)}^\mathsf{DC}\eta \Big) f_{\Dt,i}^{(n)}\triangleq g_{k,i}(\cv_{(n)}^\mathsf{AC}),
\end{align}
where $[f_{\Dt,i}^{(1)},f_{\Dt,i}^{(2)},\dots,f_{\Dt,i}^{(N_\mathsf{T})}]^\TT=\fv_{\Dt,i}$. Then, we can re-express 
\begin{equation}
    e_k = |1-v_kg_{k,k}|^2 + |v_k|^2\sum_{i\neq k} |g_{k,i}|^2 + \sigma_k^2|v_k|^2. 
\end{equation}

\subsubsection{A Closed-Form Solution to~\eqref{eq:WMMSE2}}
To solve~\eqref{eq:WMMSE2}, we can write its Lagrangian as
\begin{equation}
    \mathcal{L}(\cv_{(n)}^\mathsf{AC},\!\nu_{(n)})\! =\!\! \sum_{k=1}^K\! \beta_k(w_k e_k -\ln w_k) +\nu_{(n)}\!\big((\cv_{(n)}^\mathsf{AC})^{\!\TT}\cv_{(n)}^\mathsf{AC}-4\pi + \eta^2\big),\notag
\end{equation}
where $\nu_{(n)}$ is a Lagrange multiplier associated with the equality constraint~\eqref{eq:WMMSE2b}. According to the Karush-Kuhn-Tucker (KKT) conditions, we have 
\begin{subequations}
    \begin{align}
    \nabla_{\cv_{(n)}^\mathsf{AC,opt}} \mathcal{L} &= 0,\label{eq:0gradient}\\
    \|{\cv_{(n)}^\mathsf{AC,opt}}\|^2 &= 4\pi-\eta^2.\label{eq:equality}
\end{align}
\end{subequations}
Based on~\eqref{eq:0gradient}, one can derive
\begin{equation}
    \big(\Am+2\nu_{(n)}\mathbf{I}_{T-1}\big){\cv_{(n)}^\mathsf{AC,opt}} + \dv = \mathbf{0}, \label{eq:Ad}
\end{equation}
where
\begin{align}
    &\Am = 2\mathrm{Re}\Big\{\sum_{k=1}^K\sum_{i=1}^{K} \beta_kw_k|v_k|^2|f_{\Dt,i}^{(n)}|^2\big(\hv_{k,(n)}^{\mathsf{AC}}\big)^*\big(\hv_{k,(n)}^{\mathsf{AC}}\big)^\TT\Big\},\notag\\
    &\dv \!=\! \dv_1+\dv_2+\dv_3,\ \dv_1\!=\! -2\mathrm{Re}\Big\{\sum_{k=1}^K\beta_kw_kv_kf_{\Dt,k}^{(n)}\big(\hv_{k,(n)}^{\mathsf{AC}}\big)^*\Big\},\notag \\
    &\dv_2 \!=\! 2\mathrm{Re}\Big\{ \!\sum_{k=1}^K\! \sum_{i=1}^K\! \sum_{j=1}^K \!\beta_k w_k |v_k|^2 f_{\Dt,i}^{(j)} (f_{\Dt,i}^{(n)})^* \hv_{k,(j)}^{\mathsf{DC}} (\hv_{k,(n)}^{\mathsf{AC}})^*\eta \Big\}, \notag\\
    &\dv_3 \!=\! 2\mathrm{Re}\Big\{\!\sum_{k=1}^K\!\sum_{i=1}^K\!\sum_{j\neq n}\!\beta_k w_k |v_k|^2\! f_{\Dt,i}^{(j)}\! (f_{\Dt,i}^{(n)})^{\!*} \!(\!\hv_{k,(n)}^{\mathsf{AC}}\!)^{\!*}\! (\!\hv_{k,(j)\!}^{\mathsf{AC}})^{\!\TT}\!\cv_{(j)}^\mathsf{AC}\! \Big\}.\notag
\end{align}
From~\eqref{eq:Ad}, we obtain
\begin{equation}
    {\cv_{(n)}^\mathsf{AC,opt}} = -\big(\Am+2\nu_{(n)}\mathbf{I}_{T-1}\big)^{-1}\dv,
\end{equation}
where the Lagrange multiplier $\nu_{(n)}\in\mathbb{R}$ must be chosen to satisfy the equality constraint~\eqref{eq:equality}. It is trivial to see that multiple values of $\nu_{(n)}$ satisfying it exist, each of which is a local optimum of~\eqref{eq:WMMSE2}. While determining all local optima is challenging, the following theorem gives a practical way to determine two of these optima.
\begin{theorem}\label{theo:1}
    Let $\{\lambda_1, \lambda_2, \dots, \lambda_{T-1}\}$ denote the eigenvalues of $\Am$ in descending order. There exists a unique feasible $\nu_{(n)}$ in $(-\infty, -\frac{\lambda_1}{2})$ and another unique feasible $\nu_{(n)}$ in $(-\frac{\lambda_{T-1}}{2}, +\infty)$ that satisfies $\|{\cv_{(n)}^\mathsf{AC,opt}}\|^2 = 4\pi - \eta^2$. Moreover, the squared Euclidean norm $\|{\cv_{(n)}^\mathsf{AC,opt}}\|^2$ is a monotonic function of $\nu_{(n)}$ within both intervals $(-\infty, -\frac{\lambda_1}{2})$ and $(\frac{\lambda_{T-1}}{2}, +\infty)$. Consequently, these two local optima can be efficiently determined via bisection search.
\end{theorem}
\begin{proof}
    Since \(\Am\) is a real symmetric matrix, we can express \(\Am\) as \(\Am = \Vm \, \mathrm{diag} \{\lambda_1, \lambda_2, \dots, \lambda_{T-1}\} \Vm^\mathsf{T}\), where \(\Vm\) contains the orthonormal eigenvectors of \(\Am\). Thus, we can decompose
    \begin{equation}\notag
        \left(\Am \!+\! 2\nu_{(n)}\mathbf{I}\right)^{-1} \!\!=\! \Vm \, \mathrm{diag} \left\{ \frac{1}{\lambda_1 \!+\! 2\nu_{(n)}}, \dots, \frac{1}{\lambda_{T-1} \!+\! 2\nu_{(n)}} \right\} \Vm^\mathsf{T}.
    \end{equation}
    Then, we can express
    \begin{align}
        &\|{\cv_{(n)}^\mathsf{AC,opt}}\|^2 \label{eq:34} \\
        &= \dv^\mathsf{T} \Vm \, \mathrm{diag} \left\{ \frac{1}{(\lambda_1 + 2\nu_{(n)})^2}, \dots, \frac{1}{(\lambda_{T-1} + 2\nu_{(n)})^2} \right\} \Vm^\mathsf{T} \dv. \notag
    \end{align}
    From \eqref{eq:34}, we observe that \(\|{\cv_{(n)}^\mathsf{AC,opt}}\|^2\) is monotonically increasing with \(\nu_{(n)}\) in the interval \((- \infty, -\frac{\lambda_1}{2})\). Additionally, \(\|{\cv_{(n)}^\mathsf{AC,opt}}\|^2 = 0\) as \(\nu_{(n)} \to -\infty\), and \(\|{\cv_{(n)}^\mathsf{AC,opt}}\|^2 = +\infty\) as \(\nu_{(n)} \to -\frac{\lambda_1}{2}\). Therefore, there must exist a unique feasible \(\nu_{(n)}\) in \((- \infty, -\frac{\lambda_1}{2})\) such that \(\|{\cv_{(n)}^\mathsf{AC,opt}}\|^2 = 4\pi - \eta^2 > 0\). The same reasoning applies to the interval \((\frac{\lambda_{T-1}}{2}, +\infty)\), leading to a similar conclusion.
\end{proof}

Theorem~\ref{theo:1} enables us to identify two local optima effectively. Although these may not necessarily be the global optimum, both can be evaluated during the update of \(\cv_{(n)}\). The updated \({\cv_{(n)}^\mathsf{AC,opt}}\) is accepted only if one of the two local optima outperforms the previous \({\cv_{(n)}^\mathsf{AC,opt}}\); otherwise, the previous \({\cv_{(n)}^\mathsf{AC,opt}}\) is retained to ensure a non-decreasing iteration process.

\subsection{Solution to Step \colorcircled{RoyalBlue}{\textnormal{\footnotesize{5}}} and Projection onto the Realizable Radiation Pattern Set}
We solve the decomposition problem in Step \colorcircled{RoyalBlue}{\textnormal{\footnotesize{5}}} using the iterative algorithm proposed in~\cite[Algo.1]{Huang2024Hybrid}. Up to this point, we have obtained $\Fm_\mathsf{EM}^{\mathsf{opt}},\Fm_\mathsf{RF}^{\mathsf{opt}}$, and $\Fm_\mathsf{BB}^{\mathsf{opt}}$ through Algorithm~\ref{algo:1}. However, as mentioned earlier, the spherical harmonics-based optimization may overly refine certain radiation patterns, resulting in solutions that are not physically realizable. To address this, we introduce an additional projection step to ensure that each antenna unit selects a feasible radiation pattern. Specifically, we assume there exist $R$ available radiation patterns for each antenna unit, denoted as $\{\bar{G}_1,\bar{G}_2,\dots,\bar{G}_R\}$. Each $\bar{G}_r$ here is a function of \ac{AoD} $(\theta,\phi)$. Therefore, we make the following assignment to the gain of the $n^\text{th}$ antenna
\begin{equation}\label{eq:proj1}
    G^{(n)}\big(\theta_{k,\ell}^{(n)},\phi_{k,\ell}^{(n)}\big) \leftarrow \bar{G}_{r_n}\big(\theta_{k,\ell}^{(n)},\phi_{k,\ell}^{(n)}\big),
\end{equation}
\begin{equation}\label{eq:proj2}
    r_n = \arg\min_r\sum_{k,\ell}\Big( \bar{G}_{r}\!\big(\theta_{k,\ell}^{(n)},\phi_{k,\ell}^{(n)}\big) \!-\! \bv\big(\theta_{k,\ell}^{(n)},\phi_{k,\ell}^{(n)}\big)^{\!\TT}\!\cv^{(n)}_{\mathsf{opt}} \Big)^2\!\!,
\end{equation}
where $\mathrm{blkdiag}\big\{\cv^{(1)}_{\mathsf{opt}},\cv^{(2)}_{\mathsf{opt}},\dots,\cv^{(N_\mathsf{T})}_{\mathsf{opt}}\big\} = \Fm_{\mathsf{EM}}^{\mathsf{opt}}$. 

\section{Simulation Results}

We simulate a \ac{BS} equipped with a $3\times 3$ array of \acp{ERA}, thus $N_\mathsf{T}=9$. The system parameters are set as follows: $K=2$, $L=3$, $N_\mathsf{RF}=4$, $T=25$, $\eta=\sqrt{2\pi}$. The signal frequency is $\unit[30]{GHz}$, and the \ac{AWGN} power $\sigma_k^2=\unit[-95]{dBm},\ \forall k$. The maximum transmit power $P_\mathsf{max}=\unit[10]{dBm}$ by default. The \ac{BS} is positioned at $[0,0,10]^\TT$, while the \ac{UE} positions are randomly generated within a $\unit[200]{m}$ distance. In the projection step, this work uses the 64 radiation patterns from~\cite[Fig.~3]{Wang2025Electromagnetically}, whose physical realizability has been experimentally validated in array configurations.

\begin{figure}[t]
  \centering
  \includegraphics[width=\linewidth]{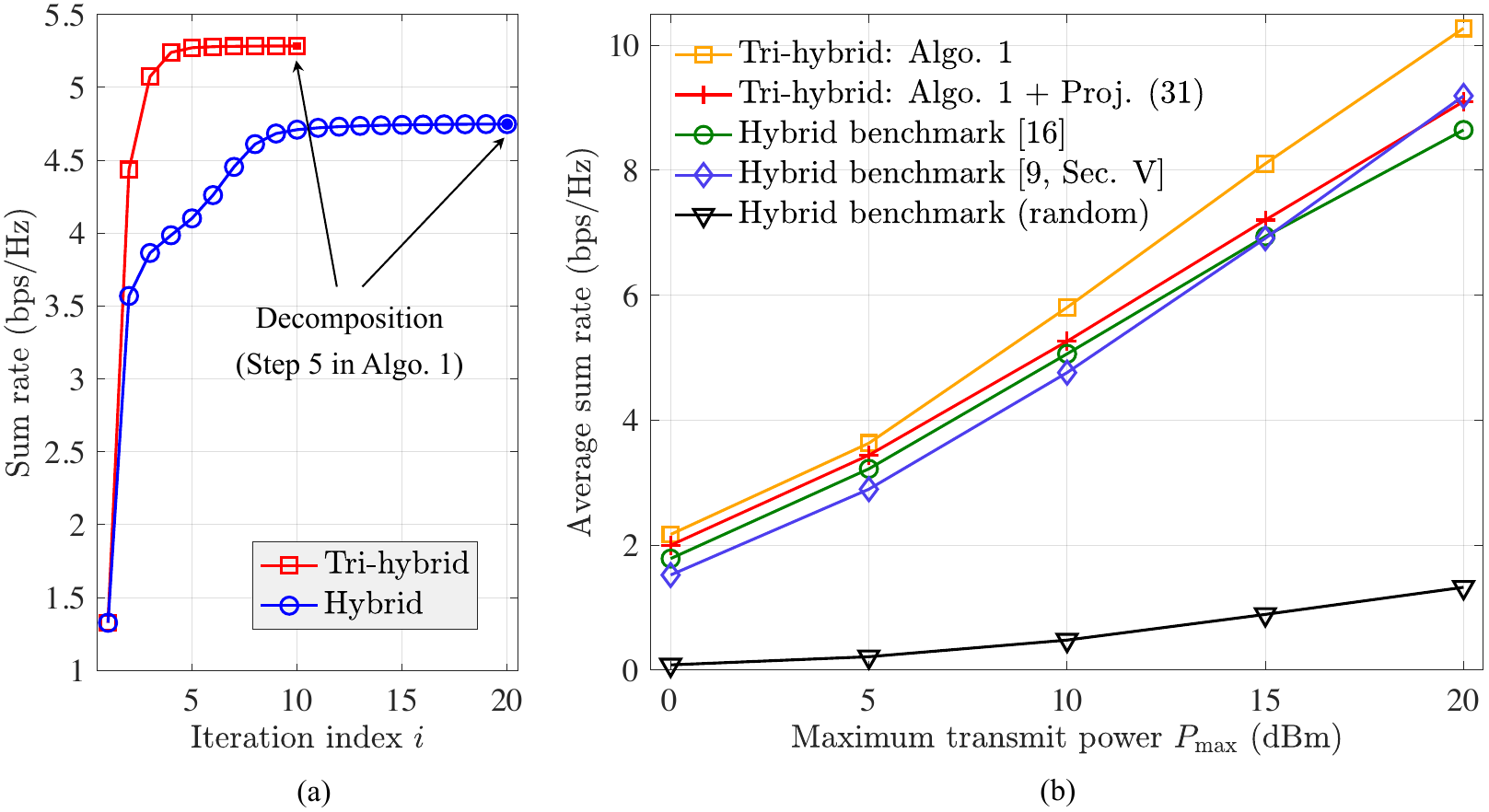}
  \vspace{-2em}
  \caption{
  Evaluation of (a) sum-rate versus iteration in Algorithm~\ref{algo:1} and (b) average sum rate versus transmit power when using various precoding approaches.}
  \label{fig_sims}
\end{figure} 

Figure~\ref{fig_sims}-(a) illustrates the sum rate as a function of iteration 
$i$ in a single run of Algorithm~\ref{algo:1}. The results show that the proposed tri-hybrid precoding solution not only achieves a higher sum rate but also converges faster compared to the conventional hybrid precoding approach. Here, the conventional hybrid precoding is implemented using the WMMSE solution~\cite{Shi2011Iteratively}, which corresponds to running Algorithm~\ref{algo:1} without updating $\Fm_\mathsf{EM}$. Furthermore, Fig.~\ref{fig_sims}-(b) evaluates the average sum rate as a function of the maximum transmit power \( P_\mathsf{max} \), comparing our approach with two benchmark methods from~\cite{Shi2011Iteratively} and~\cite{Sohrabi2016Hybrid}. The results indicate that optimizing the radiation pattern of each antenna in the spherical harmonic space without considering realizability allows the tri-hybrid precoding scheme to achieve a significantly higher sum rate than conventional hybrid precoding architectures. However, after projecting onto the set of realizable radiation patterns, the performance gain diminishes considerably. In some cases, such as $P_\mathsf{max}=\unit[20]{dBm}$, the projected tri-hybrid solution even underperforms compared to conventional hybrid solutions. This suggests that the adopted projection steps~\eqref{eq:proj1} and~\eqref{eq:proj2} are suboptimal, highlighting the need for improved methods to identify the optimal realizable radiation patterns to fully unlock the potential of \acp{ERA}.

\section{Conslusion}
We studied the tri-hybrid multi-user precoding problem based on the radiation pattern-reconfigurable \acp{ERA}, where the digital precoder, analog precoder, and \ac{EM} precoder are optimized jointly. The problem is tackled through an alternating optimization approach, with each subproblem solvable in closed form. Additionally, we performed a projection step that projects the optimized radiation pattern onto a set of realizable radiation patterns obtained from a real \ac{ERA} prototype. The results indicate that when accounting for real-world realizability, the performance enhancement may not be as significant as the theoretically optimized results in spherical harmonic space. This limitation can be overcome by developing more effective projection methods and/or improving \ac{ERA} hardware design.

\bibliography{references}
\bibliographystyle{IEEEtran}

\end{document}